\documentclass[oneside,11pt]{article}
\usepackage[headheight=1in,margin=1.1in]{geometry}
\usepackage[utf8]{inputenc} 
\usepackage[T1]{fontenc}
\usepackage{lipsum}
\usepackage{amsfonts}
\usepackage{graphicx}
\usepackage{epstopdf}
\usepackage{natbib, bbm, pdfpages}
\usepackage{amssymb}
\usepackage{fancyhdr}
\usepackage{amsthm}
\usepackage{algorithmic}
\usepackage[hidelinks=true]{hyperref}
\usepackage[symbol]{footmisc}
\usepackage{afterpage} % To defer the table placement

\usepackage{bm}
\setcitestyle{round}

\newtheorem{proposition}{Proposition}

\usepackage[acronym]{glossaries}
\newacronym{ind}{ind.}{individual}
\newacronym{hh}{hh.}{household}
\newacronym{inds}{inds.}{individuals}
\newacronym{hhs}{hhs.}{households}

\usepackage{rotating}
\usepackage{tablefootnote}

\usepackage{geometry}
\usepackage{pdflscape}
\usepackage[raggedrightboxes]{ragged2e}
\usepackage[capitalize,nameinlink]{cleveref}[0.19]

\renewcommand{\vec}[1]{\mathbf{#1}}

\newcommand{\A}{\boldsymbol{A}}
\newcommand{\B}{\boldsymbol{B}}

\newcommand{\ER}{Erd\H os-R\'enyi}

\newcommand{\acl}[1]{\acrlong{#1}}

\usepackage{algorithm}[0.1]

\newcommand{\abs}[1]{\left|#1\right|}

\usepackage{soul}

\usepackage{todonotes}

\title{\vspace{-1.5cm}The illusion of households as entities in social networks}

\author{Izabel Aguiar\\ Santa Fe Institute \\ \texttt{izabel@santafe.edu} 
      \and
      Philip S. Chodrow\\Middlebury College\\  \texttt{pchodrow@middlebury.edu}
      \and
      Johan Ugander\\ Stanford University\\ \texttt{jugander@stanford.edu}}

\date{}
\begin{document} 
\maketitle

\begin{abstract}
Data recording connections between people in communities and villages are collected and analyzed in various ways, most often as either networks of individuals or as networks of households. These two networks can differ in substantial ways. The methodological choice of \emph{which} network to study, therefore, is an important aspect in both study design and data analysis. In this work, we consider various key differences between household and individual social network structure, and ways in which the networks cannot be used interchangeably. In addition to formalizing the choices for representing each network, we explore the consequences of how the results of social network analysis change depending on the choice between studying the individual and household network---from determining whether networks are assortative or disassortative to the ranking of influence-maximizing nodes. As our main contribution, we draw upon related work to propose a set of systematic recommendations for determining the relevant network representation to study. Our recommendations include assessing a series of \textit{entitativity criteria} and relating these criteria to theories and observations about patterns and norms in social dynamics at the household level: notably, how information spreads within households and how power structures and gender roles affect this spread. We draw upon the definition of an \textit{illusion of entitativity} to identify cases wherein grouping people into households does not satisfy these criteria or adequately represent given cultural or experimental contexts. Given the widespread use of social network data for studying communities, there is broad impact in understanding which network to study and the consequences of that decision. We hope that this work gives guidance to practitioners and researchers collecting and studying social network data.
\end{abstract}

\section{Introduction} \label{sec:intro}
Social networks represent how people are connected to one another, for example, through physical contact, friendship, or the lending and borrowing of goods. In many studies of social connections in communities, towns, or villages, social networks are typically collected and analyzed at one of two representational levels: individual or household. That is, when considering how people are connected with one another, we can consider which \textit{individuals} are connected to one another or which \textit{households} are connected to one another. There is therefore an important choice to be made in any study of such a social network: should the nodes represent individuals or households, and how should the relationships between the nodes be represented? In this work we focus on emphasizing and formalizing this choice, exploring possible consequences of it, and systematizing recommendations for how to make it.

Often \citep[e.g.,][]{banerjee2013, airoldi2024}, \textit{household network} data is collected by first collecting the \textit{individual network}, and then grouping those individuals (and their connections) together with those who live in the same household (see \Cref{fig:hh_ind_diagram}). When household networks are constructed in this way, two key assumptions are implicitly made: (i) that an individual's social connections are shared and can be utilized by every other individual in their household, and (ii) that households are connected through individual relationships. In practice, the impacts of these implicit assumptions are not considered and thus individually-driven processes and relationships are conflated with household ones. Given that both the household and individual networks represent the same \textit{idea}\textemdash how the same set of people are connected to one another\textemdash the distinctions between household and individual networks is nuanced. 

\begin{figure}
    \centering
    \includegraphics[width=0.8\linewidth]{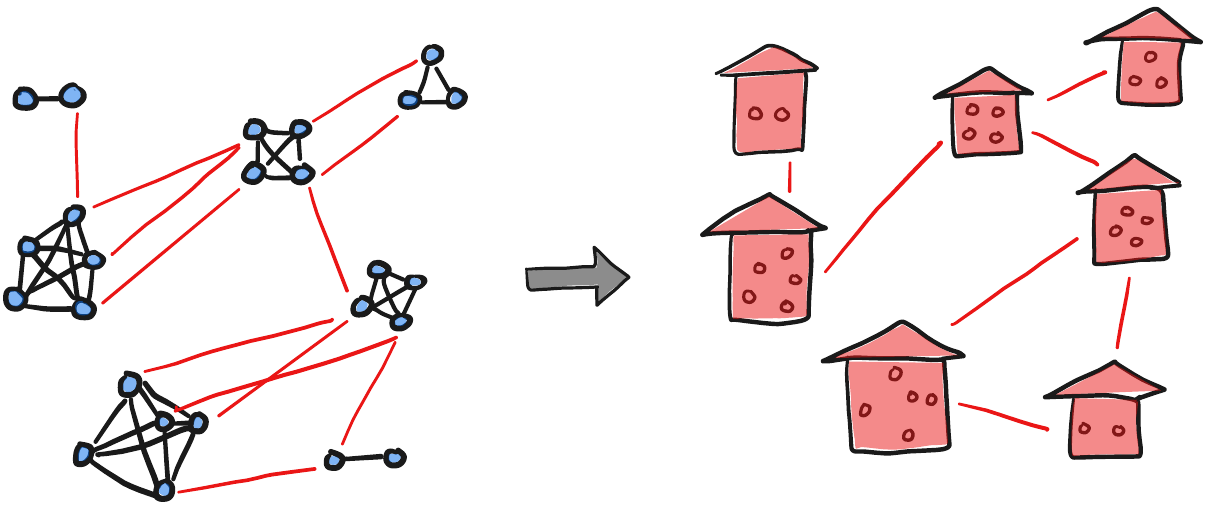}
    \caption{This work considers the choice of which network is the most appropriate to study in a given context, a choice which we present as consequential for meaningful empirical network analysis. In this diagram we show how individual networks (left, blue) are often translated to household networks (right, red). In what we refer to throughout as the \textit{individual network}, nodes are individuals and an edge between two individuals is collected through a survey. In what we refer to as the \textit{household network}, nodes are households and an edge between two households is usually determined by aggregating the relationships between individuals in that household. We specify this decision to represent household edges in this way as the \textit{basic household contraction rule} and propose alternate methods for defining edges between households in \Cref{subsubsec:cont_rules}. In the diagram here, we also represent individuals within the same household as completely connected to one another, an aspect of some individual network datasets which we discuss in more detail in \Cref{subsub:local}. We review how individual and household networks are collected and studied in practice in \Cref{tab:review}.}
    \label{fig:hh_ind_diagram}
\end{figure}

\begin{figure}
    \centering
    \includegraphics[width=.95\linewidth]{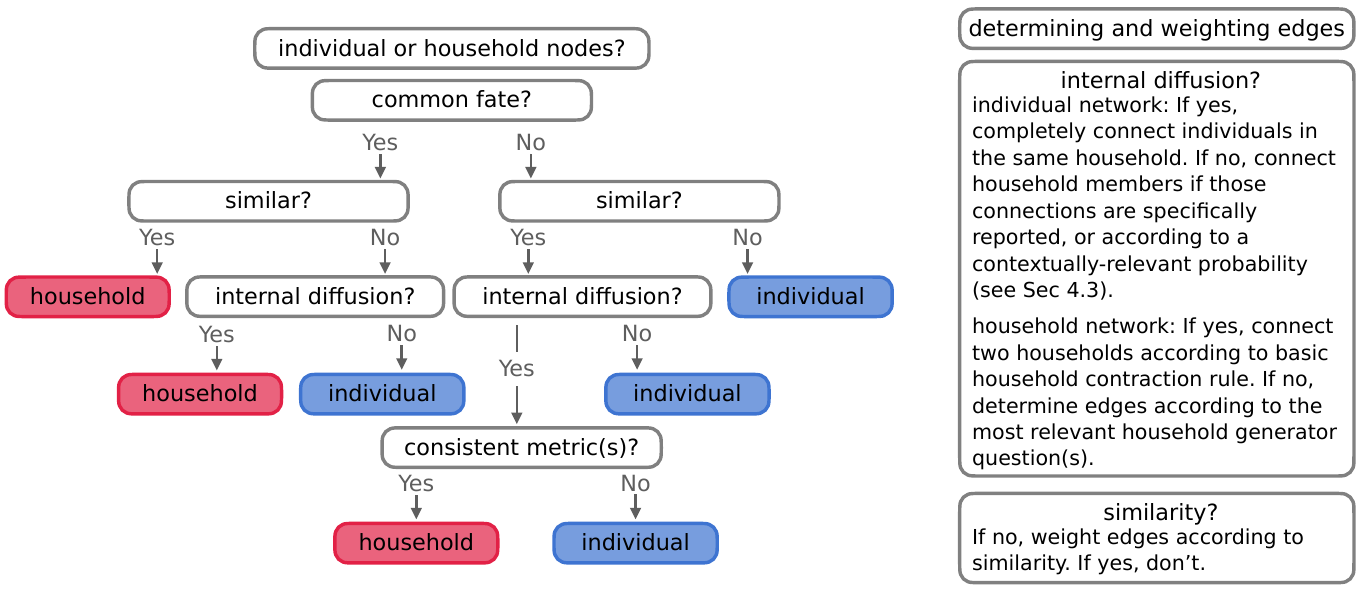}
    \caption{In this work we provide a systematic recommendation for determining whether the household or individual network should be studied given a particular context and experimental goal or intervention. The decision tree here poses a contextual evaluation of a set of entitativity criteria\textemdash \textit{proximity}, \textit{similarity}, \textit{common fate}, and \textit{internal diffusion}, which we discuss in detail in \Cref{sec:rec}\textemdash to determine an appropriate level of node aggregation, as well as to suggest how to weight edges. In \Cref{fig:ex_decision} we apply these recommendations to three separate examples to determine an appropriate network to analyze.}
    \label{fig:dec_chart}
\end{figure}

Plenty of social science research ranging from within sociology \citep{campbell1958common}, political science \citep{nickerson2008, auld2013inter, cheema2023canvassing}, anthropology \citep{werner1998, schmink1984household, niehof2011conceptualizing, koster2018family}, and network science \citep{kumar2024friendship} have theorized, measured, and observed that there are assumptions with consequences when studying individuals and their relationships interchangeably with their respective aggregates.

In this work, we contribute a bridging of insights and observations from disparate fields in an effort to systematize recommendations for researchers collecting, studying, or developing methods for social network data. We review and adopt insights from experimental work testing how political messages spread within households and detail relevant observations from ethnographic studies of household interactions. We adapt \citeauthor{campbell1958common}'s (\citeyear{campbell1958common}) theoretical \textit{entitativity criteria} for determining when a group of individuals can be reasonably studied as an aggregate, or when representing individuals as one, coherent household is inadequate in a given setting\textemdash an \textit{illusion}. By incorporating work from a broad range of disciplines, we can help researchers in making a rigorous choice for how to align their research question with an appropriate social network, so that their subsequent analysis and observations may be consistent with their research goals.

We explore the differences between household and individual networks through two approaches: through formalizing the choices in their representations and considering the consequences of those choices. To formalize the decision, we introduce notation for household and individual networks with their respective adjacency matrices and \textit{contraction rules} that map one to the other. We complement this decision with a discussion of how the individual and household networks can be further specified using weighted edges, including based on gendered connections, or more generally, the similarity of the nodes in specific contexts. 

We examine how different choices in this formalization can lead to substantially different conclusions based on various network metrics. Theoretically, we show a simple example showing that \textit{heterogeneous random node aggregations} in an \ER~random graph result in a network that is no longer \ER. We then consider how local metrics, the spread of information over a  network, and centrality metrics substantively differ on the individual and household village social support networks from \cite{banerjee2013}. Notably, we explore how analysis\textemdash and the corresponding conclusions about the social network\textemdash considerably differs between the individual and household networks, highlighting the importance of choosing the appropriate network to study in the context of a given research aim.

To assess when an individual or household network should be studied in a given context, we provide a systematic recommendation based on a series of \textit{entitativity criteria} \citep{campbell1958common}. We ground these recommendations in theories and experimental observations studying how individuals interact within households and as aggregates, as well as how gender and power interact in household networks to distinguish between the types of connections most relevant in a given setting \citep{werner1998}. We propose specific adaptations and extensions to the \textit{entitativity criteria} of \textit{proximity}, \textit{similarity}, \textit{common fate}, and \textit{internal diffusion} as they relate to studying networks in the context of interventions and experimental goals, which we organize as a decision tree in \Cref{fig:dec_chart}. In doing so, we relate the entitativity criteria to recommendations on how and when to collect and weight the relationships between individuals or households differently. To highlight examples of how to evaluate the criteria in different settings, we apply this set of recommendations to three different large-scale experimental network studies, \cite{banerjee2013}, \cite{alexander2022algorithms}, and \cite{airoldi2024}.

The structure of this work is as follows. 
In \Cref{sec:background} we briefly discuss the related work that we both build upon and directly use to make recommendations. In Sections \ref{sec:choices}, we formalize a set of possible ways to represent the household network given information about the individual network. Next, in \Cref{sec:consequences} we explore the impact of choosing either the household, individual, or a gendered network when studying a variety of different network metrics in empirical networks. In \Cref{sec:rec} we propose a set of systematic recommendations for how practitioners can decide which network is the most appropriate to study. We conclude in \Cref{sec:conclusion} with concluding remarks and suggestions for future work.

\section{Related Work}\label{sec:background}
The study of how individuals within social networks act and react in aggregate broadly spans political science, anthropology, and sociology. We briefly introduce this related work so that we may adapt it to make concrete recommendations the context of social networks in \Cref{sec:rec}. We note that the related work we consider here often spans or defies disciplinary boundaries. As considered in the network science literature, the abstract question we consider (of node-aggregation) is closely related to the literature on \textit{coarse-graining} \citep{itzkovitz2005coarse, gfeller2007spectral, kim2004geographical, klein2020emergence} where the goal is to explore \textit{how} to aggregate nodes while preserving some property of the network (e.g., degree distribution, spectral properties, properties of random walks). However, as we will see in this section, questions of how \textit{people and their social interactions} can be considered in aggregates have been considered in many disparate literatures.

\subsection*{Name generators and multilayer networks}\label{subsec:namegen}
The question of how individuals are related through social networks has been extensively studied in the literature on \textit{name generators} \citep{campbell1991name}, which specifically studies the social network survey question(s) being asked and the impact the choice of question has on the structure of the resulting network. In the present work, we echo the scholarship's contextual focus on \textit{which relationships we care about} and \textit{how to ask the relevant question(s)} to capture those relationships. Furthermore, we complement the line of work by considering \textit{which nodes} (individuals or households) we care about representing in a given context. Here we also see the connection between the present work, name generators, and \textit{multilayer networks} \citep[e.g.,][]{kivela2014multilayer}. Succinctly, multilayer networks can capture relationships given by multiple name generator questions, by representing each \textit{type} of relationship in a different layer. In the context of household networks, multilayer networks allow for the possibility of multiple types of household connections\textemdash whether gendered, neighbourly, individual, or shared\textemdash to be represented.

\subsection*{Epidemiology}
Within epidemiology there has been a focus on developing variations of epidemiological models to account for differences in transmission rates between individuals living within the same household, and how accounting for households and higher-order interactions changes the predicted spread of a disease \citep{becker1995effect, fraser2007estimating, st2021social, boccaletti2023structure}. This literature contributes not only specific modeling choices to account for the impact of households on disease transmission, but also makes more general graph-theoretical contributions. In \Cref{subsec:inf_dif} we draw upon the modeling distinctions from this body of work in how we model the spread of information between individuals within and beyond those living in the same household. In contrast to this work, we consider how household networks are constructed, the implications of treating households as social entities in social (non-biological) processes, and under what conditions it is appropriate to do so.

\subsection*{Diffusion of voter information campaigns within the household}\label{subsec:gotv} 
How information spreads and is shared within households, particularly in the context of political messaging, has been extensively experimentally tested. 
In \citet{nickerson2008}, the authors target households with two registered voters and deliver a ``Get out the Vote'' message to whichever household member answers the door, ultimately finding that 60\% of the propensity to vote is passed to the other member of the household. In considering how this finding might be impacted by gender differences, \cite{cheema2023canvassing} studies the difference between targeting women and men with voter information interventions in Pakistan. They find that female voter turnout increases by 5.4 and 8.0 percentage points, respectively, when either the man or both man and woman are targeted, as compared with only targeting the woman in a household. \citet{ferrali2022registers}, however, finds that household heads (regardless of their gender) have a strong influence over whether other adults in their household register to vote. \citet{bhatti2017voter} complicates this and related findings by suggesting that, rather than the \textit{information} from voter education campaigns being shared with other household members, the spillover of increased voter participation can be attributed to social pressure exerted by the targeted household member. That is, their work suggests that pure information does not diffuse within a household, but rather personal opinions and social pressure.

These experimental findings suggest that information shared with one member of a household does not definitively imply that all members of the household will receive or be impacted by it \citep{nickerson2008}. Moreover, the \textit{impact} of information on a household can be vastly different, depending on the gender and hierarchical position of \textit{who} in the household receives it \citep{cheema2023canvassing, ferrali2022registers}. Furthermore, observing an outcome does not distinguish between information and social pressure diffusing between household members \citep{bhatti2017voter}. 

Broadly, this body of work cautions against an assumption that information reaching one member of a household guarantees that it reaches every other household member, an assumption implicit in how household networks are constructed in practice. Indeed, \textit{why}, \textit{how}, and \textit{how much} information is shared between household members is contextually dependent on hierarchy within the household, gender norms and expectations, and the type of information being shared. 

\subsection*{Household strategies and household decision making}\label{subsec:housestrat}

The rich literature on \textit{household strategies}\textemdash which studies \textit{household networks} as forms of interaction and informal means of economic exchange\textemdash specifically distinguishes between individual relationships and household relationships. An observation key to our work is explicitly stated in \citeauthor{werner1998}'s ethnographic work on household strategies in Kazakhstan \citep{werner1998},  wherein she observes that it is ``wrong to assume that a household network equals the sum total of all household members' social bonds\dots In reality, a social bond established by one household member may or may not be at the disposal of other household members.'' Not all relationships that an individual has are accessible to all members of their household\textemdash the accessibility of particular connections may be dependent on context, gender, and power. A salient example of this claim is presented by \cite{werner1998}: a wife does not have direct access to her husband's professional connections, and her access and benefit of them must be mediated through him. 

Indeed, gender and power are intimately connected with \textit{how} household relationships are created and maintained, and the broader role of household networks is culturally, politically, and economically dependent \citep[see, e.g.,][]{werner1998, schmink1984household, wallace2002household, niehof2011conceptualizing, yotebieng2018household}. 

These ethnographic and theoretical findings are echoed in economic theory, where work shows that there is an economic advantage for a division of labor amongst members in a household \citep{becker1993treatise, lundberg2008gender}. This economic literature on \textit{household decision-making} departs from a more traditional economic model of the household as a single unit. By accounting for differences in the individuals within a household, they show that it is more efficient for the utility of a household as a whole if individuals specialize their attention and labor in either market or household decisions. Whereas the divide in decision-making responsibilities may be determined by gender, when there is any difference in the aptitudes, incentives, or interests of individuals within a household, specialization increases the efficiency of the household as a whole. These findings provide theoretical grounding that decisions, even if made at the household level, may be the primary responsibility of only one member of the household.

\subsection*{Entitativity}\label{subsec:entity}

In her ethnographic evaluation of household strategies, \cite{werner1998} remarks that, ``individuals are social actors, but households are not.'' Questioning \textit{when} a group of people can be considered as a social entity itself leads us to the foundational work of \cite{campbell1958common}, an in-depth theoretical assessment of this very question. \citeauthor{campbell1958common} looks towards the biological sciences to propose several ``metrics'' for determining what he terms \textit{entitativity}\footnote{In reference to this clunky word, we echo \citeauthor{campbell1958common} in apologizing that ``the present writer regrets adding two suffix syllables to the word \textit{entitative}, already three fourths suffixes.''}: the degree to which a group of people can be considered one entity.\footnote{We note that \citeauthor{campbell1958common} exclusively uses the word ``metric'' colloquially with respect to rules for assessing whether a group is an entity or not, and not paired with any quantitative sense of a metric. To avoid confusion with quantitative metrics discussed in the present work, we instead use the word \textit{criteria} for Campbell's notion of ``metric.''}

The several criteria \citeauthor{campbell1958common} proposes are those assessing the \textit{proximity}, \textit{similarity}, \textit{common fate}, \textit{internal diffusion}, and \textit{reflection or resistance to intrusion} of the individuals within a potential aggregate. We consider the first four criteria in more detail and within the context of household networks in \Cref{sec:rec}, and discuss the fifth criterion as an interesting direction for future work in \Cref{sec:conclusion}.

Importantly, \citeauthor{campbell1958common} cautions against what he calls \textit{illusions} of entitativity which occur when solely relying on \textit{one} criteria to determine when an aggregate of people can be considered an entity. He comments on the importance of confirming boundaries of entitativity with multiple criteria, writing that \textit{illusions often rely on superficial boundaries.} It is in this sense that we aim to diagnose when it is an \textit{illusion} to treat households as social entities: if, in a given context, the boundary of the household unit is not validated by other entitativity criteria.

We adapt the findings of this related work into a systematic recommendation for determining when the household or individual network should be collected and analyzed in a given setting. We organize our recommendations in  \Cref{fig:dec_chart}, where we propose a decision tree for determining what network to study in a given context. We discuss the evaluation criteria in detail in \Cref{sec:rec}.

\section{Choices} \label{sec:choices}
In this section, we first formalize the \textit{choice} to be made when studying a social network\textemdash broadly, to represent the network at either the individual or household level. The choice, in the language of networks, is twofold: what should the nodes represent, and what constitutes an edge between two nodes? In this work, we focus particularly on how these two questions are answered in the context of \textit{individual} and \textit{household} networks, although one could instead choose to consider aggregating individuals according to other affiliations, like club memberships, dorm rooms, or hobby groups. 

If the choice is to study the individual network, \textit{how} to represent edges may be as straightforward as relying on the union of multiple different name generator questions. However, determining how to represent a network of household relationships may be more complicated. We first consider how to represent the household network through different \textit{node contractions} on the individual network, and then turn to consider the choice to capture household relationships by relying on specific \textit{household generator} questions.

\subsection{Node contraction} \label{subsubsec:cont_rules}
Node contraction \citep{oxley2006matroid, bollobas2013modern} describes the graph operation when two nodes $v_i$ and $v_j$ in a graph are replaced by one node $w$, and every edge incident to either $v_i$ or $v_j$ becomes incident to $w$. As we reference more generally in this work, a \textit{set} of nodes can be contracted to one node $w$, and every edge incident to any node in the set becomes incident to $w$. 

\paragraph{Basic household contraction.} To make the above statement more precise, in the most basic case, for a graph $G = (V, E)$, where $V$ is a set of nodes and $E$ a set of edges between them, consider a set of node-contraction sets $H=\{H_1, H_2, \dots, H_R \}$ where each $H_r$ is a non-overlapping partition of $V$ such that $\sum_{r}^{R} \abs {H_r}  = N$, so each $v_k \in V$ belongs to some set $H_r$. Then the graph $G'$ which results from contracting these vertices is given by
\begin{equation} \label{eq:contraction}
    \begin{aligned}
    &G' = (V', E')\\
    &V' = \{ w_i \}_{i=1}^{R} \\
    &E' = \left\{ (w_i, w_j) \mid \exists \, v_k \in H_i \textrm{ and } v_\ell \in H_j \textrm{ such that } (v_k, v_\ell) \in E \right\}.
    \end{aligned}
\end{equation}
We note that if the original graph is weighted, directed, undirected, or contains loops, the graph resulting after contraction can be of the same type. For a simple graph, we extend the last condition above to be for $i \neq j$. 

In the literature, simple, unweighted household networks (either directed or undirected) are often constructed by contracting individual networks according to \cref{eq:contraction}\textemdash each individual is replaced by one node, and all of the edges incident to any individual become incident to the household node. However, we note that different \textit{types} of household networks can be constructed through node contraction rules. 

Below we define additional possible node contraction rules that, given an \acl{ind} network $G = (V, E)$ and a set of households and their corresponding members, $H = \{ H_1, H_2, \dots, H_R\}$, produce a \acl{hh} network $G' = (V', E')$. We note that there are many more contraction rules and that many of the below rules can be combined to define additional rules (e.g., one can define a \acl{hh} network by combining weighted and gendered contraction rules from below).

\paragraph{Weighted household contraction.} 
For node contraction on a weighted graph $G= (V, E, W)$, the weight of an edge, denoted $W(w_i, w_j)$ is equal to sum of the weights of the edges incident to any node within each corresponding contraction set. 
\begin{equation*} \label{eq:weighted_con}
    \begin{aligned}
    &G' = (V', E', W')\\
    &V' = \{ w_i \}_{i=1}^{R} \\
    &E' = \{ (w_i, w_j) \mid \exists \, v_k \in H_i \textrm{ and } v_\ell \in H_j \textrm{ such that } (v_k, v_\ell) \in E \},\\
    &W'(w_i, w_j) = \sum_{v_k \in H_i} \sum_{v_\ell \in H_j} W(v_k, v_\ell). 
\end{aligned}
\end{equation*}
The weight of an edge can also be normalized according to the proportion of \acl{inds} in $H_i$ who are connected to any \acl{ind} in $H_j$, capturing a general strength of the edge between the two households. 

\paragraph{Gendered household contraction.} For specified gender $a$, we can adapt \Cref{eq:contraction} such that an edge from household $H_i$ to $H_j$ exists if there is an edge from an \acl{ind} of gender $a$ in $H_i$ to an \acl{ind} of gender $a$ in $H_j$: 
\begin{equation*} \label{eq:gendered_con}
    \begin{aligned}
        E' = \{ (w_i, w_j) \mid \exists \, v_k \in H_i \textrm{ and } v_\ell \in H_j \textrm{ such that } (v_k, v_\ell) \in E    \textrm{ and } gender(v_k)=gender(v_\ell) = a\}.\\
    \end{aligned}
\end{equation*}
By considering that the presence of an edge between two households (or the strength of an edge) may be dependent on the gendered relationships between individuals in each household, we allow for the flexibility to account for potential gendered household labor, gendered sharing of information, and gendered access to relationships that has been observed and theorized in, e.g., \cite{werner1998}, \cite{berti2015adequacy}, and \cite{wallace2002household}. More generally, gender here can be replaced by other contextually-relevant identity markers. Looking ahead, we consider how gendered household networks impact the centrality of households in \Cref{subsec:gender}, and we consider how gendered differences and differences in power (with respect to household decision-making) relate to the \textit{similarity} of a set of nodes in a given context in \Cref{sec:rec} and show examples in \Cref{fig:ex_decision}. 

\paragraph{Household generator questions and layered household contraction.} 
Household networks are typically studied by aggregating information from the individual network level using the basic households rule in \Cref{eq:contraction} (see \Cref{tab:review}). However, the \textit{household strategies} literature we briefly reviewed in \Cref{sec:background} highlights that there is a key difference between individual relationships and those which entire households utilize, benefit from, and maintain. The contraction rules previously discussed are most suitable in contexts wherein it is appropriate to construct the household network \textit{from} the individual network\textemdash as opposed to constructing the household network from particular \textit{household generating questions}. 

This distinction is highlighted by contrasting the implications of two common name generator questions, \textit{``Who would you borrow kerosene or rice from?''} and \textit{``Who do you go to for medical advice?''} \citep[e.g.,][]{banerjee2013}. The first question may generate names of individuals who are stand-ins for an entire \textit{household} (if you go to X's house to ask for rice but X is not home, do you ask whoever answers the door?). The latter question, however, generates a particular \textit{individual} network, distinct from a \textit{shared household} one (the trust involved in asking X for advice about your health may not transfer to X's spouse). Furthermore, how respondents interpret and answer these survey questions is further complicated by the ambiguity as to whether ``you'' is implied to signify the individual survey respondent or the household they represent. That is, in addition to questions soliciting nominees which represent an entire household, in some circumstances the respondent may be answering \textit{on behalf of} their entire household. This ambiguity is dependent on the culture, language, and setting in which the question is being asked.

We can consider that when the individual social network is collected using the union of $P$ different name generators it can also be represented as a \textit{multilayer network} $\mathcal{G} = \{G_1, G_2, \dots, G_P\}$ \citep[see, e.g.,][]{kivela2014multilayer}. As such, we may construct the household network through node contraction on the set of the most relevant household generating layer(s), $\mathcal{L}$, where
$G_L = (V_L, E_L)$ for each $L \in \mathcal{L}$:
\begin{equation*} \label{eq:layered_con}
    \begin{aligned}
    &G' = (V', E')\\
    &V' = \{ w_i \}_{i=1}^{R} \\
    &E' = \{ (w_i, w_j) \mid \exists \, (v_k, v_\ell) \in E_L \textrm{ for some }L \in \mathcal{L} \textrm{ and for } v_k \in H_i, v_\ell \in H_j  \}.
    \end{aligned}
\end{equation*}
That is, if one can contextually distinguish certain name generators as soliciting relationships between households, a household network can be reasonably constructed by only connecting households which are connected in those particular layers.

\subsection{Choices in practice} \label{subsubsec:emp_choice}
To show how the choice to collect, represent, and study a social network at either the individual or household level is made in practice, we review a wide range of empirical studies on social networks in \Cref{tab:review}. We see that social networks are often collected by asking \textit{individuals} name generator questions about their social support systems (sometimes about multiple types of social support). Individuals and their connections are then grouped together with other individuals living in the same household to make a \textit{household} social network, through the ``basic household contraction'' rule discussed above, forming a household network that is the ultimate object of study. 

In reviewing each study we focus on identifying (i) whether the \textit{intervention} or \textit{experimental design} of the study occurs at an individual or household level, (ii) if the social network data is collected by asking individual or household questions, and (iii) if the network studied is the individual or household network. We delineate these characteristics to motivate our work and clearly highlight opportunities for rigorously aligning a research question with its corresponding social network data collection and analysis. In \Cref{sec:rec} and \Cref{fig:ex_decision} we focus on three relatively recent large scale experiments, \cite{banerjee2013}, \cite{banerjee2019}, and \cite{airoldi2024}, which we review in greater detail in \Cref{apx:emp_rev}.

%%%%%%%%%%%%%%%%%%%%%%%%%%%%%%%%%%%%%%%%
%          LANDSCAPE TABLE             %
%%%%%%%%%%%%%%%%%%%%%%%%%%%%%%%%%%%%%%%%
\afterpage{%
  \clearpage% Ensure the table starts on a fresh page
  \newgeometry{margin=.75cm} 
  \begin{landscape}
    \begin{table}
\begin{tabular}{|p{1.2in}|p{.9in}|p{1in}|p{2.2in}|p{4in}|}\hline
\vspace{.05cm}\textbf{Paper} &\vspace{.05cm} \textbf{Network collected} &\vspace{.015cm} \textbf{Network studied} &\vspace{.05cm} \textbf{Context} &\vspace{.05cm} \textbf{Intervention}\\ \hline
  \vspace{.5pt} \cite{alatas2012targeting} & \vspace{.5pt}  Household\tablefootnote{The household network in \cite{alatas2012targeting} is collected by interviewing one individual from each household. The person is asked to list other households they are related to and all social groups that each person in their household belongs to. To construct the household network, an edge connects any two households with individuals in the same social group.}& \vspace{.5pt}  Household & \vspace{.5pt}  The adoption of an agricultural technology across 631 villages in Indonesia & \vspace{.5pt} An individual is targeted with information about the agricultural technology. \\ \hline
 \vspace{.5pt}\cite{banerjee2013} &\vspace{.5pt}  Individual &\vspace{.5pt}  Household\tablefootnote{\cite{banerjee2013} notes the motivation for using the household network is because loan-participation occurs at the household level.} &\vspace{.5pt} Adoption of microfinance loans across 72 villages in Karnakata, India. & \vspace{.5pt} An individual is targeted with information about a microfinance loan \\ \hline
 \vspace{.5pt} \cite{cai2015social} & \vspace{.5pt} Individual\tablefootnote{In \cite{cai2015social} the head of each household is asked to list their five closest friends. Individuals who are \textit{not} head of their household can be included in the individual network only if they are nominated, and may not nominate their own friends. Notably, in this study heads of households are ``almost exclusively male.''} &  \vspace{.5pt} Household &\vspace{.5pt} Adoption of weather insurance across 185 villages in rural China. &\vspace{.5pt}  The head of a randomly targeted household participates in an informational session about the weather insurance product. \\ \hline
 \vspace{.5pt} \cite{chami2017social} &\vspace{.5pt} Individual &\vspace{.5pt} Household\tablefootnote{\cite{chami2017social} also notes the motivation behind their decision to study the household network: ``because community medicine distributors (CMDs) were trained to and have been shown to move from door to door to deliver medicines during MDA.''} &\vspace{.5pt} Fragmentation of the contact network by targeting households with deworming treatments across 17 villages in Uganda. &\vspace{.5pt} With different targeting techniques, a random household is visited by a medical team to administer a deworming treatment. The household is considered noncompliant if \textit{any} individual within the household refuses treatment. \\ \hline
\vspace{.5pt} \cite{banerjee2019} &\vspace{.5pt} Individual &\vspace{.5pt} Household &\vspace{.5pt} The spread of information in 71 villages in Karnakata, India. &\vspace{.5pt} The household corresponding to an individual selected either randomly or by community nomination is given information about a non-competitive raffle.\tablefootnote{In a second phase of experiments in \cite{banerjee2019} on a different set of villages, the information is about vaccination clinics.}\\ \hline
\vspace{.5pt} \cite{alexander2022algorithms} &\vspace{.5pt} Individual &\vspace{.5pt} Individual (representing household) &\vspace{.5pt} Spillover effects of a health intervention under different seeding strategies in 50 residential buildings in Mumbai, India. &\vspace{.5pt} A subset of women are given information about the health benefits of iron-fortified salt and are given coupons to buy the salt and share with their networks. \\ \hline
\vspace{.5pt} \cite{airoldi2024} &\vspace{.5pt} Individual &\vspace{.5pt} Household &\vspace{.5pt} Spillover effects of education about maternal health in 176 rural villages in Honduras. &\vspace{.5pt} A healthworker comes to a random household at regular intervals over the course of 22 months to educate household members about maternal health. \\ \hline
\end{tabular}
\caption{An overview of major social network experiments in networks with household structure. For each study we categorize the level the network is collected and studied at, the level of the intervention, and additional context.}\label{tab:review}
\vspace{-1.5in}
\end{table}
 % Include your table here
  \end{landscape}
  \restoregeometry
  \clearpage% Ensure the subsequent text starts on a new page
}
%%%%%%%%%%%%%%%%%%%%%%%%%%%%%%%%%%%%%%%%
%      End LANDSCAPE TABLE             %
%%%%%%%%%%%%%%%%%%%%%%%%%%%%%%%%%%%%%%%%

\section{Consequences} \label{sec:consequences}
We now turn to explore the consequences that this choice might lead to. First, we consider a simple case of random node contraction on a random graph to theoretically explore how node aggregation affects graph structure. We then explore how aggregating individuals into household nodes on empirical networks can impact (i) local network metrics, (ii) the spread of information over a network and influence maximization, and (iii) node centrality.

In relation to the closely related work on \textit{coarse-graining}, which we mention in \Cref{sec:background}, recent work by \cite{griebenow2019finding} explores how different network properties (e.g., average degree, eigenvector centrality, betweenness) change on simulated networks after aggregating nodes according to a proposed spectral algorithm. The consequences we explore in this section explore how the network properties change under \textit{household-defined node aggregation}.

\subsection{Random graphs with random households}\label{subsub:ER}

To understand how node contraction can impact the structure of the resulting graph, we first explore an example of considering the effect of randomly contracting nodes in a random graph, formalizing a result regarding expected infectious contacts between households noted in \cite{doenges2024sir}. Consider a graph on $n$ nodes where an edge exists between any two nodes with independent probability $p$, called an \ER~random graph and denoted $G(n, p)$. In \Cref{apx:ER} we show that when households are of the exact same size, the resulting graph can be explained by a \ER~generative process, but with fewer nodes and with a different edge probability. This observation, however, no longer holds when we consider households of differing sizes.

\begin{proposition}
   Consider a random graph $G(V,E)$ generated by a $G(n,p)$. A corresponding unweighted household graph $G'(V',E')$ is constructed by constructing $m$ disjoint node sets $H_1, H_2, \dots, H_R$ of size $\ell_1, \ell_2, \dots \ell_R$ by choosing nodes uniformly at random, and contracting these sets into corresponding nodes $V' = \{w_r\}_{r=1}^R$, where $n = \sum_{r=1}^R\ell_r$. Then $G'$ is no longer described by an \ER~model, but an \text{inhomogeneous \ER~model} \citep{bollobas2007phase} where edges do not exist with the same probability across nodes. Instead, an edge $(w_q, w_r)$ exists independently with probability $1 - (1-p)^{\ell_q \ell_r }.$ The expected degree of node $w_k$ is $\mathbb{E}[deg(w_k)] = \sum_{j = 1}^R 1 - (1-p)^{\ell_k \ell_j}.$ For a simple graph, $j \neq k.$ 
\end{proposition}

\begin{proof}
    In the graph $G$, for any pair of nodes $v_i, v_j \in V$, the edge $(v_i, v_j)$ does not exist with independent probability $1-p$. Now consider node sets $H_q$ and $H_r$. Then for node $v_i \in H_q$, there is probability $(1-p)^{\ell_r}$ that there is no edge between node $v_i$ and any node in set $H_r$. Since there are $\ell_q$ nodes in set $H_q$, and each edge probability is independent, the total probability of no edge existing between any node in $H_q$ and any node in $H_r$ is $(1-p)^{\ell_q \ell_r }$. Thus for the graph $G'$ where any node $w_k$ is constructed by contracting nodes in set $H_k$, then for any pair of nodes $w_q$ and $w_r$ in $G'$, edge $(w_q, w_r)$ exists independently with probability $1 - (1-p)^{\ell_q \ell_r }.$ Thus we can describe $G'$ as a realization of an inhomogeneous \ER~model where the expected degree of node $w_k$ is $\mathbb{E}[deg(w_k)] = \sum_{j = 1}^R 1 - (1-p)^{\ell_k \ell_j}.$ If $G$ and $G'$ are simple graphs, we enforce $j \neq k$. 
\end{proof}

The point of this example is to show that contracting nodes can significantly alter the structure of the resulting graph, by changing the variation in expected degrees across nodes in the network. With this example in mind, we turn to analyze a set of empirical networks to see how metrics change under household node contraction.
\subsection{Local network properties}\label{subsub:local}

For the empirical analyses which follow, we consider the collection of 72 networks of social support in southern India from \citet{banerjee2013}. In \cite{banerjee2013}, data is collected from all individuals in a household about other individuals in the village they get and give social support to and from. Then, this individual-level network data is aggregated at the household level according to the \textit{basic household contraction} rule we describe above. In the following sections, we thus compare metrics and experiments when considering the individual network or the household network. See \Cref{apx:emp_rev} for more details about the data and the contexts in which it was collected. 

We begin our empirical exploration by considering how two commonly-used local network metrics, degree assortativity and average clustering coefficient, are substantively different depending on whether they are measured on the individual or household network (see \Cref{fig:deg_clus}). Notably, the average clustering coefficient is relatively small when measured on the household networks, and relatively large on the individual networks. Given that clustering coefficient can be used to describe the social health of a community \citep[e.g.,][]{cartwright1956,bearman2004}, we see in the case of these social support networks that staggeringly different conclusions might be drawn depending on whether the individual or household network is analyzed. This further highlights the importance of choosing the most contextually relevant network for understanding and evaluating a community using network methods.

Similarly, over the same set of villages, networks aggregated at the household level largely display disassortative mixing, whereas the corresponding individual networks have positive degree assortativity. This observation echoes the work of \cite{newman2002assortative} and \cite{newman2003social} and the authors' explorations into differences in degree assortativity in social and biological networks. However, in their work they show that social networks are largely assortative and biological networks are largely disassortative. It is thus surprising, in the context of the \citeauthor{banerjee2013} networks, to find both assortativity and disassortativity on the same social networks represented at different granular representations. The assortativity of a network has implications for how a disease is sustained in a community and the resilience (or lack thereof) of a network to the removal of nodes \citep{newman2003social}. As such, we again see that distinct conclusions are drawn regarding these communities, depending on which network is considered further indicating that \textit{the choice of which network to study} has significant implications. 

A similarly peculiar observation to the above was recently made in \cite{kumar2024friendship}: that when analyzing the same village networks \citep{banerjee2013}, their proposed graph property \textit{inversity} (closely related to degree assortativity) was substantively different depending on whether or not it was measured on the household or individual network. In their work, \citeauthor{kumar2024friendship}~propose inversity as a correlation-based metric to characterize when different seeding strategies will be more effective based on the structure of the network\textemdash a positive inversity implies one strategy, and negative another. For more in-depth details, see \cite{kumar2024friendship}.

The surprising aspect of their finding, then, is that, although each network captures the same people living in the same village, different seeding strategies are determined to be more impactful if one analyzes the network at either the household or individual level. Notably, \citeauthor{kumar2024friendship}~write that the result of this finding is that one should consider different strategies dependent on whether or not the intervention is expected to occur at the individual or household level. We expand on their recommendation with particular principled evaluation criteria in \Cref{fig:dec_chart} and \Cref{sec:rec}.

\begin{figure}
    \centering
    \includegraphics[width=0.95\linewidth]{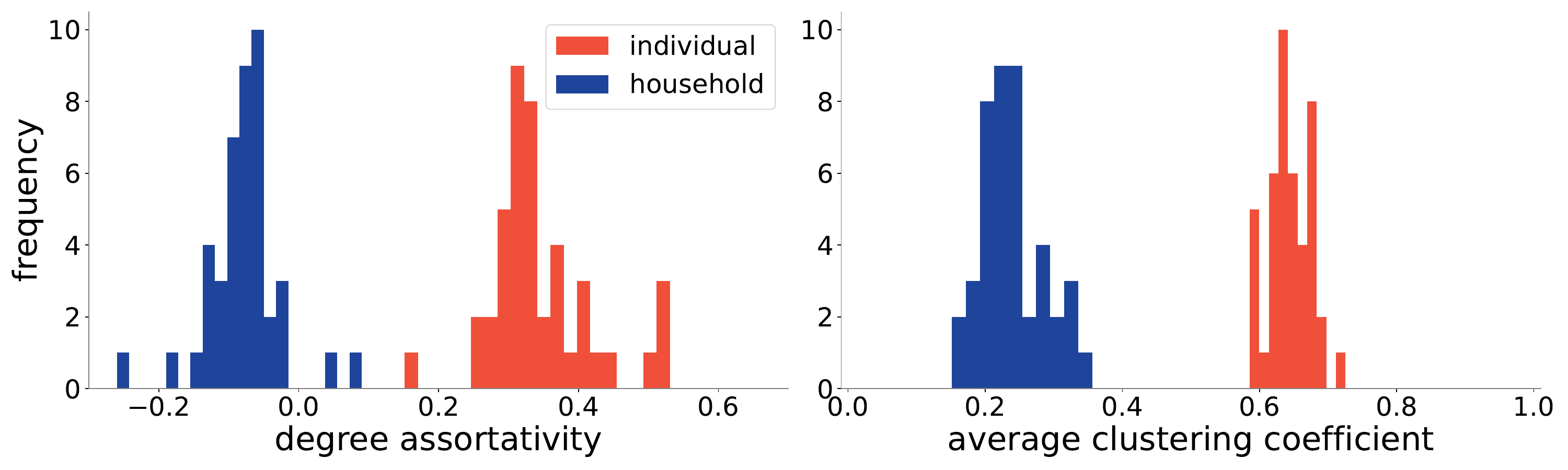}
    \caption{The household and individual networks from \cite{banerjee2013} have substantively different interpretations when considering both the degree assortativity and clustering coefficients.}
    \label{fig:deg_clus}
\end{figure}

To explore what might be causing the shift in inversity in these networks, we notice that within each household, all individuals within a household are completely connected to all other individuals in their household. That is, the network of individuals is essentially made up of many different cliques, loosely connected to other cliques (cliques amongst household members is also how we represent the individual network in \Cref{fig:hh_ind_diagram}). Surprisingly, if we remove all of the edges between individuals in the same household, we see that inversity measured on the individual networks more closely aligns with inversity on the household networks. In fact, we see that for the majority of the villages, the corresponding networks both have positive inversity (see \Cref{fig:inv_nohh}). This observation suggests that the flip in inversity that \citeauthor{kumar2024friendship}~observe can be attributed to the existence of cliques within households. We explore this reasoning further alongside a more in-depth discussion of inversity and additional numerical experiments in \Cref{apx:inversity}.

\begin{figure}
    \centering
    \includegraphics[width=0.8\linewidth]{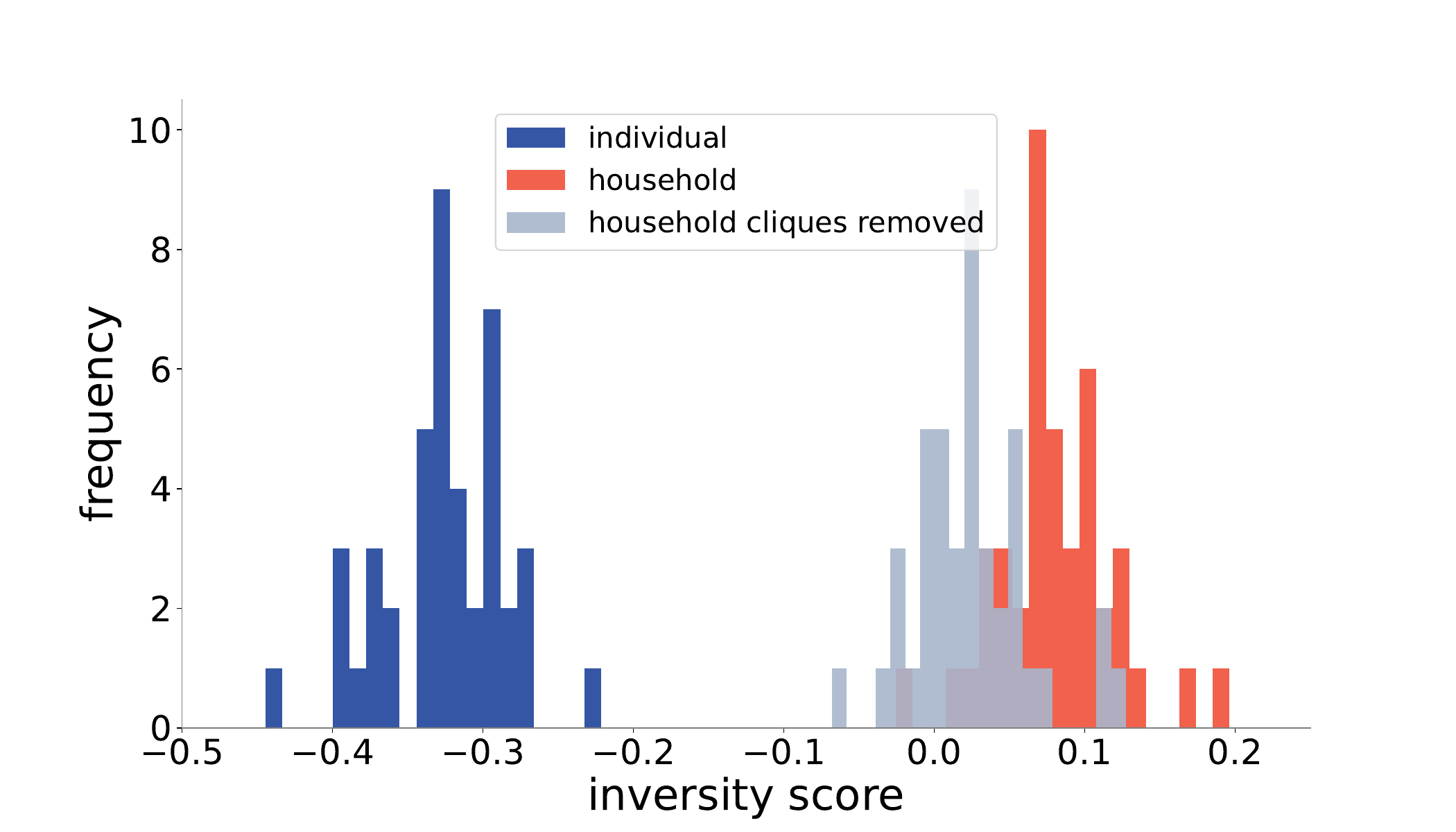}
    \caption{The inversity of the \cite{banerjee2013} networks at the individual %(blue) 
    and household %(red) 
    level. In light grey, we plot the histogram of the inversity of the individual networks across villages when the intrahousehold edges are removed.}
    \label{fig:inv_nohh}
\end{figure}

\subsection{Spread of information over a network}\label{subsec:inf_dif}

We turn now to explore what differences we uncover if the spread of information is modeled on either the individual or household network. To address this question, we conduct an experiment on the individual and household networks from \cite{banerjee2013} to identify which nodes are \textit{influence maximizing}: which nodes when seeded with a piece of information will spread it to a maximal proportion of the network. The problem of influence maximization is relevant for efficiently and effectively distributing a limited resource with maximal impact on the network. On each network, we aim to identify the seed set of $K=10$ nodes which maximize the spread of information using the greedy strategy proposed in \cite{kempe2003maximizing}. 

To construct the set of influence maximizing nodes, we first run $1000$ \textit{independent cascade} simulations\textemdash where the spread of information is modeled as traveling across an edge with probability $p$\textemdash with each separate node as a single starting seed. We add to the seed set a single node which spreads information to the highest proportion of the network, on average, across the $1000$ simulations. We then iterate this process until the seed set has $K=10$ nodes, adding the node which, conditional on the nodes already in the seed set, reaches the highest proportion of nodes in the network across the simulations. \cite{kempe2003maximizing} showed that this greedy algorithm for assembling a seed set for an independent cascade is guaranteed to produce a set that reaches a population within a factor of $(1-1/e)$ of the optimal seed set (while also showing that the underlying optimization problem itself is NP-hard).

To relax the idea of completely connecting households in the individual network, we propose decomposing the individual adjacency matrix into the parts made up of \textit{intrahousehold} and \textit{extrahousehold} edges: edges which connect individuals in the same and in different households, respectively. 
\begin{align*}
    \A = \A_{extra} + \A_{intra},
\end{align*}
For individual nodes $i$ and $j$ belonging to households $H_k$ and $H_\ell$, respectively, then $\A_{extra}$ and $\A_{intra}$ are defined as,
\begin{align*}
    A_{extra, \, ij} = \begin{cases}
        0 \textrm{ for } H_k = H_\ell\\
        1 \textrm{ otherwise}
    \end{cases} A_{intra, \, ij} = \begin{cases}
        0 \textrm{ for } H_k \neq H_\ell\\
        1 \textrm{ otherwise.}
    \end{cases}
\end{align*}
Decomposing the individual adjacency matrix in this way allows us to consider the impact and appropriateness of intrahousehold edges in different contexts. Specifically, this decomposition allows us to vary the weight we give to intrahousehold connections:
\begin{align*}
    \A^*_p = \A_{extra} + (1- p) \A_{intra},
\end{align*}
for $p \in [0,1]$. Here, $p$ can allow for flexibility in the presence of relationships between people in the same household, by varying the probability that an intrahousehold connection exists in a given context. For example, small $p$ represents a high probability that information, once given to one household member, will spread to the other members of the household. By contrast, large $p$ would capture a situation where there is low probability of information spreading between household members. 

By making this distinction, we build upon the modeling recommendations of \cite{newman2003social}, wherein they consider how to project bipartite affiliation networks onto a corresponding social network. In the context of modeling the spread of information, writing adjacency matrices in this way can represent a rate of the spread of information within households and make a direct modeling connection to the work of, e.g., \cite{nickerson2008} who found that information does not deterministically reach all household members. Furthermore, our modeling choice to allow for variation between the probability of information spreading between individuals within and beyond households is reflected in the epidemiological work considering how household transmission affects the spread of a disease \citep[e.g.,][]{ball2012sir}.

In modeling the independent cascade on both individual and household networks we fix the probability of information passing between two households to be the same. For the purpose of our experiments, we model this probability with $q = 0.05$.\footnote{We chose this extrahousehold information passing probability to be relatively low so that information didn't completely saturate the networks in our experiments. In preliminary experiments with higher values of $q$, we found an even bigger difference between influence-maximizing seed sets, most probably because when information saturates the network there are many equivalently influence-maximizing sets, and comparison between sets becomes arbitrary.} For the individual networks we assume information passes along intrahousehold and extrahousehold edges at different rates: whereas we keep the probability of information passing between extrahousehold edges as $q=0.05$, we model the intrahousehold probability as $(1-p)=0.7$. Note that with $(1-p)=1$, independent cascade unfolds in the exact same way on the individual and household networks, and with $(1-p) = 0$ information can only pass along edges between individuals in different households. As such we choose $(1-p)=0.7$ to model an intermediate scenario, where information spreads within households, but not deterministically. This setup follows observations from, e.g., \cite{nickerson2008} and \cite{cheema2023canvassing} whereby intrahousehold transmission is high but not guaranteed (see \Cref{sec:background}). 

We construct the influence-maximizing seed sets for both the individual and household networks, $S_i$ and $S_h$, respectively, for each village. To compare the sets of nodes across individual and household networks, we map each seed in the individual sets to their corresponding households. In \Cref{fig:jac} we report the overlap between these pairs of sets across all of the villages. We see that for many of the villages, there is very little overlap between the sets of influence-maximizing households. This observation implies that a different set of households (or individuals) will be identified as influence-maximizing seeds, dependent on whether or not the household or individual network is studied. 

\begin{figure}
    \centering
    \includegraphics[width=0.6\linewidth]{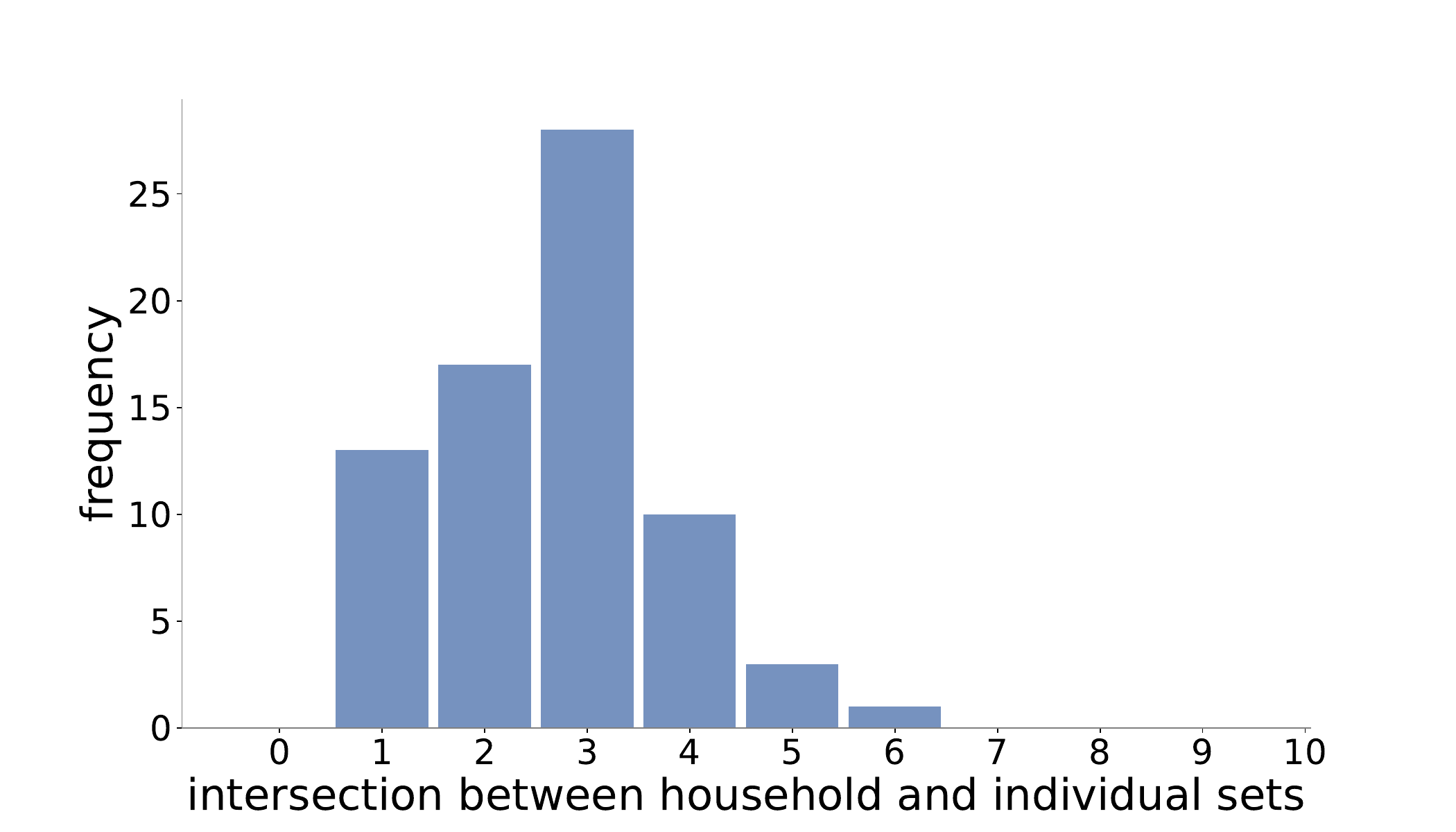}
    \caption{The intersection between the sets of the 10 most influential households, found by greedily maximizing the average proportion of nodes reached over 1000 independent cascades on the household and individual networks from \cite{banerjee2013}. To compare sets, we map the set of the 10 influence maximizing individuals to their corresponding households.}
    \label{fig:jac}
\end{figure}

However, this observation does not imply that the \textit{impact} of choosing different seeds will be vastly different. There may be many local optima of very similar quality and furthermore, because we use a greedy algorithm for constructing both $S_i$ and $S_h$, the differences in output from this algorithm need not even necessarily reflect differences in the true optimal sets. Thus, we inspect the difference between the expected proportion of the network reached over 100 iterations of an independent cascade on both the household and individual networks, using either $S_i$ or $S_h$ as the seed set. 

When running the independent cascade on the household network with $S_i$ as the seed set, we directly map the individuals to their corresponding households to get the corresponding set of nodes on the household network, $S_{i\rightarrow h}$. When running the independent cascade on the individual network with $S_h$ as the seed set, we map the households in $S_h$ to a set of corresponding individuals, which we call $S_{h\rightarrow i}$. To do so, we map households to a random subset of the individuals who live in the household (where the subset is selected uniformly at random with $(1-p)=0.7$ to match our numerical experiments on the individual network). This mapping means that $S_{h\rightarrow i}$ is a much larger set than $S_i$. We then consider the difference between the expected proportion of the network reached when running 1000 independent cascades on either the household or individual network, using $S_i$ and $S_h$ (or their corresponding mappings, $S_{i\rightarrow h}$ and $S_{h \rightarrow i}$, respectively) as seed sets.

We find that the difference in the expected proportion of nodes reached by the independent cascade only varies by at most around 4\% of the network in both the experiments on the individual and household network. However, we notice that the difference is greater when we simulate the spread of information on the individual network. 

This observation provides perspective for the results of our analysis above: although the \textit{sets} of the most influential nodes are different depending on whether or not they are selected from the individual or household network, the \textit{impact} of this difference depends on what the set is being used for. If the set is identified to seed information in hopes that it will have maximal influence, then there may not be a significant difference between the proportion of the network reached. If the set is used to compare influence-maximizing nodes with a set of nodes chosen in some other way (e.g., according to their status as village leaders), then the conclusions will depend on the choice of which network they are selected from.

\subsection{Diffusion centrality and gendered networks}\label{subsec:gender}
As we discussed in \Cref{sec:background} and \Cref{sec:choices}, the relationships represented in networks of households can be supported, maintained, and connected by gendered labor and gendered relationships. Moreover, as we consider in \Cref{sec:rec}, some experimental frameworks consider the spread of information that is more relevant and accessible to certain genders. 

As a particular example of gendered networks, \cite{alexander2022algorithms} considers the network of women living in a shared apartment building in Mumbai, India. The researchers tell a subset women about a particular type of fortified salt available alongside its health benefits and track how information about the salt travels throughout the network and which non-targeted women end up buying the product. The work is an excellent example where the researchers collected and studied a network that represented the process and context of the experiment. We see that in this network representation, women\textemdash as the household members responsible for buying groceries\textemdash stand-in for their entire household, and the network only connects them to other women with the same responsibilities. As we will see in \Cref{subsec:rec_ex}, when we apply the recommendations we propose in \Cref{fig:dec_chart} to this experimental setting, we ultimately recommend studying the same network (see \Cref{fig:ex_decision}). 

In this section we consider the consequences of representing a household network with gendered edges, as we formalized in \Cref{eq:gendered_con}. Moving from contagion to centrality, we consider how the diffusion centrality \citep{banerjee2013, banerjee2019} of nodes changes depending on whether or not we weight household edges according to the gendered relationships between individuals in each household. 

Diffusion centrality, a variation of eigenvector centrality, as defined for directed networks in \cite{banerjee2019}, can be interpreted as how extensively a piece of information spreads as a function of an initially informed node. Consider a directed, strongly connected network with corresponding weighted adjacency matrix $\boldsymbol{w}$ where $w_{ij}$ gives the relative probability that $i$ shares information with $j$. Then for a piece of information originating at node $i$, in each time period an informed node tells each neighbour $j$ the information with independent probability $w_{ij}$. If we consider that this process happens for $0 < T < \infty $ time periods, diffusion centrality is specifically defined as:
\begin{align}
    DC(\boldsymbol{w}, T):= \left(\sum_{t=1}^T (\boldsymbol{w})^t \right) \boldsymbol{1}.
\end{align}

To explore how the diffusion centrality of a household depends on whether or not we only consider gendered connections, we first subset the household networks from \cite{banerjee2013} to just the surveyed households, for which we have information about the gender of the individuals in each household. We find the diffusion centrality of each household in (the connected component of) the subsetted household network: $DC(\B, T) \vec{1} := DC_{B}$, where $\B$ is the adjacency matrix corresponding to the household network $G'$.
Next, we create two separate gendered networks: one where two households are connected only if a female from one house is connected to a female of the other house, $G'_{FF}$, and another where two households are connected only if a male from one house is connected to a male in the other house, $G'_{MM}$. We then find the diffusion centrality for each household in these networks, $DC_{B_{FF}}$ and $DC_{B_{MM}}$, respectively. 

Similar to the sensitivity analysis that \cite{banerjee2013} conduct \citep[see Table S7 in][]{banerjee2013} with respect to changes in $T$, we investigate how the use of gendered networks impacts the resulting diffusion centrality of each household by finding the correlations between the diffusion centralities, $DC_{B}, DC_{B_{FF}}$, and $DC_{B_{MM}}$. We visualize the Pearson correlation coefficient between each pair of vectors across all 72 villages in \Cref{fig:gender_corrs}. Notably, we see that the diffusion centrality on the male-to-male network is strongly positively correlated the diffusion centrality on the entire household network, implying that the spread of information on the household network is dominated by household relationships defined by male-to-male individual relationships. The diffusion centrality of households connected by considering only female-to-female relationships is also positively correlated with the household diffusion centrality, although we observe much more spread across villages. This observation indicates that the household network defined without considering gender does not necessarily capture the most relevant network, especially in the case where information may be more likely to spread across female-to-female connections. 

Overall we see that the households with the highest diffusion centrality vary across these gendered networks, highlighting the importance of understanding \textit{who} a piece of information is relevant to \citep[e.g., the women in each household who buy groceries, see][]{alexander2022algorithms} and \textit{what social norms} allow for communication about particular subject matter to take place \citep[e.g., sharing information about maternal health may be stigmatized as only relevant and shareable among women, see][]{berti2015adequacy}. Note that for the purpose of the experiments in this section we only considered weighting the household networks with possible extremes, zero or one, by either including an edge or not depending on the gender of individual connections between those households. However, this weight can be relaxed in contexts where information sharing is less gendered, such that, for example, household edges defined by only only male-to-male individual connections are weighted lower than household edges defined by female-to-female individual connections but are not completely absent.

Thus we conclude this section with another instance of how conclusions and interventions can meaningfully vary (in this case, by misidentifying the most central households in a village) depending on the choices of how to represent a network. We see that understanding how gender (or, as we discuss in \Cref{sec:rec}, other relevant similarities between individuals) plays a role in an experimental context is important for properly considering how people and their households are connected. 

\begin{figure}
    \centering
    \includegraphics[width=0.6\linewidth]{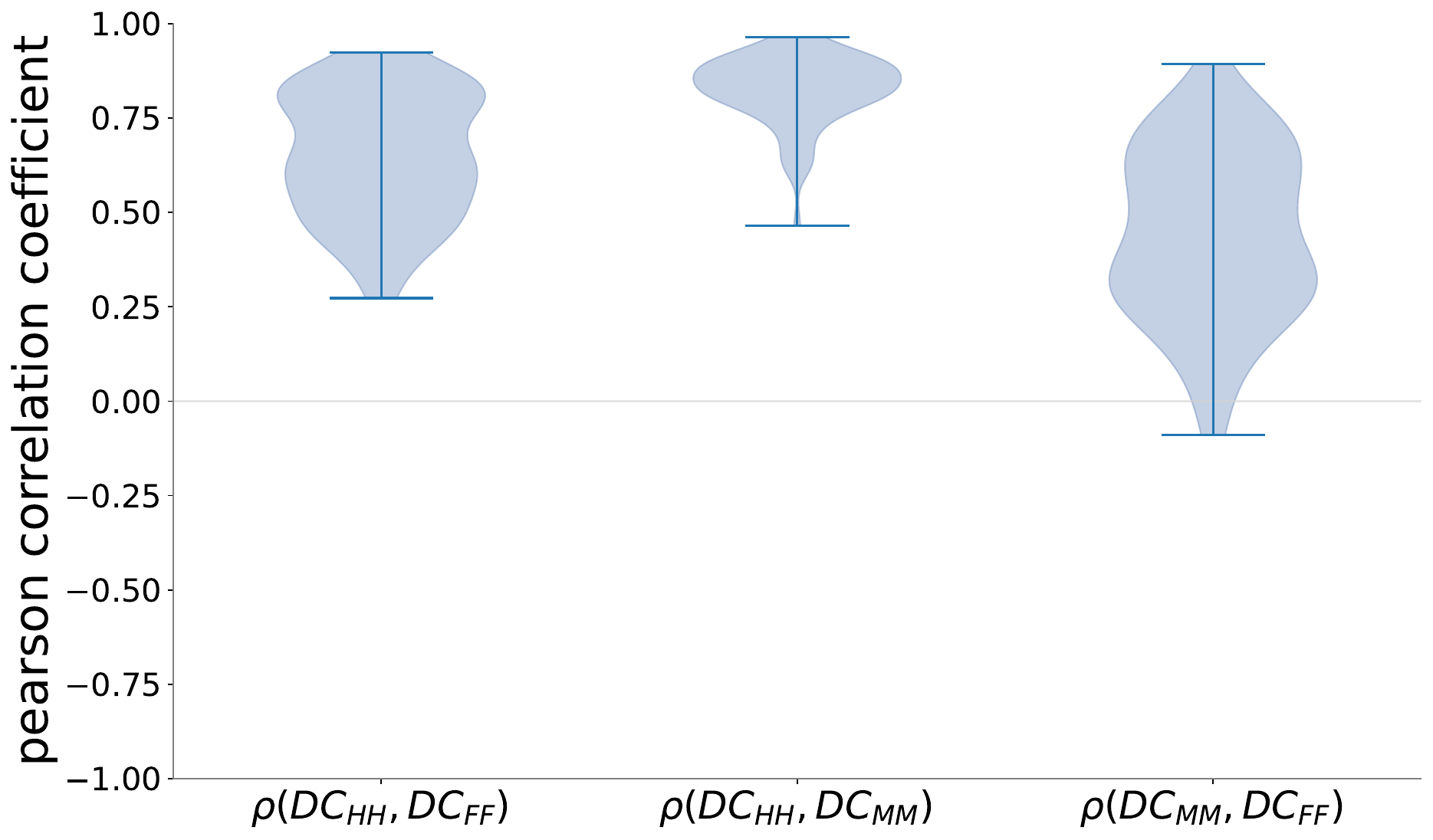}
    \caption{The correlations between each pair of diffusion centrality vectors across all 72 villages. We observe that the diffusion centrality of households varies across these gendered networks. A village with high/low correlation corresponds to households having similar/different centralities across different network definitions.}
    \label{fig:gender_corrs}
\end{figure}

\section{Recommendations} \label{sec:rec}
As we saw in \Cref{sec:consequences}, the insights from a given network analysis can substantively vary between individual and household networks. In order to help guide the decision of which network should be analyzed in a certain context, in this section we discuss specific recommendations based on the work we reviewed in \Cref{sec:background}. 

\subsection{Entitativity criteria} \label{subsec:entitativity_rec}
We first adapt \citeauthor{campbell1958common}'s criteria for entitativity\textemdash \textit{proximity}, \textit{similarity}, \textit{common fate}, and \textit{internal diffusion}\textemdash specifically to their potential use in assessing households as social entities. We also propose a dependent hierarchy of these entitativity criteria in \Cref{fig:dec_chart}, where we propose a decision tree for determining what network to study in a given context and how to weight edges or include intrahousehold edges. We note that while we do not adapt \citeauthor{campbell1958common}'s criterion of \textit{resistance to intrusion} here, we discuss doing so as a path for interesting future work in \Cref{sec:conclusion}.

\paragraph{Proximity.}
Proximity of individuals is determined by their physical closeness to one another. Although it may seem straightforward to assess this criteria for individuals in a household (because they all live in the same house), in some places and on certain timescales, household members may actually not be proximate. For example, assessing proximity must be approached with care in places where it’s common for one or more household members to leave the village/town for work for extended periods of time \citep[see][]{niehof2011conceptualizing}, where the boundary of the household is more culturally nebulous \citep{yotebieng2018household}, or when household and family composition changes in longitudinal studies \citep{steele2019modeling}. Determining the boundaries of what defines a household, who is in a household, and whether household members are proximate in an experimental context depends both on the timescale of the intervention and on cultural norms regarding household membership. These considerations aside, an aggregate of individuals living in the same household necessarily meet this criterion. Assessing if the household meets the following entitativity criteria \textit{alongside} proximity considers whether or not the household boundary is what \citeauthor{campbell1958common} calls an \textit{illusion}\textemdash a boundary that satisfies one criterion but not others.

\paragraph{Similarity.} Similarity of individuals within their household unit should be particularly assessed with respect to an experimental question and/or intervention. Similarity here not only refers to a relevancy and accessibility of the intervention with respect to different household members' identities \citep{berti2015adequacy}, but also to their roles and responsibilities within the household \citep{becker1993treatise}. For example, if the experiment is focused on the spread of information, similarity must be assessed with respect to the content of the information that the individuals are exposed to and if the individuals' differences in gender, power, or household responsibilities impact the sharing of information. Guiding questions to assess this criterion in a given context might be: 
\begin{itemize}
    \item \textit{Is there a gender-related norm, stigma, or significance regarding sharing this type of information?}
    \item \textit{Does one person in the household make financial decisions on behalf of their family that might be relevant to this experimental context?}
    \item \textit{Does the subject or outcome of the experiment lie primarily within the specialization of one household member?}
\end{itemize}

\paragraph{Common fate.} This criterion addresses, within the specific context of the intervention and tracked outcomes for a particular study, whether or not individuals living in the same household share the same outcome. If the outcome and/or intervention occurs at an individual level, the individuals likely do not share a common fate. In \citeauthor{campbell1958common}'s introduction of this criteria he alludes to particles within a rock\textemdash if the rock is thrown, all the particles in the rock will be thrown together, and thus share a common fate. Indeed, the concept of common, or \textit{linked} fate in sociology has a long history, particularly in the context of considering racial, ethnic, and socioeconomic groups \citep{brewer2000superordinate, simien2005race}. In the context of individual and household networks, we consider individuals within a household as sharing a common fate when a decision or intervention implicates all of the household members (e.g., the decisions to take on a microfinance loan, install a new roof, or purchase a new household grocery item are all decisions in which the individuals within the same household share a common fate).

\paragraph{Internal diffusion.} This criterion considers whether information (or a behaviour, or a disease) spreads to all individuals in a household once one household member is exposed to it. This criterion is particularly relevant and important when assessing entitativity of households in experiments concerned with the spread of information (or disease dynamics), because modeling contagion on the household network necessarily implies that once a household is exposed to information, \textit{all individuals} within the household are as well. 

If a household network is the most appropriate network to study in a given scenario, particular attention should be given to \textit{which} questions are generating the network, as well as \textit{who} is being asked the questions. As we propose in \Cref{fig:dec_chart}, it is most important to capture \textit{household} connections as distinct from the aggregate of individual connections when there is no internal diffusion amongst household members. In this case, it is important to capture only the relationships between individuals which represent a broader connection between the households as a whole. We discussed this distinction in the discussion of \textit{household generator questions and layered household contraction} in \Cref{sec:choices}, wherein we pointed out the difference between the significance of the name generators, \textit{``Who do you borrow kerosene or rice from?''} and \textit{``Who do you go to for medical advice?''} Whereas the first question is still asked of \textit{individuals}, it represents a larger relationship between \textit{households}, and thus becomes the relevant relationship to consider in this scenario. 

\paragraph{Consistent metrics.} We add this criterion to \citeauthor{campbell1958common}'s set of criteria for use in the particular setting of analyzing networks, and is helpful in cases where determining entitativity might be non-obvious. This criterion is meant to assess whether or not the particular network metric of interest (e.g., an influence-maximizing seed set, clustering coefficient, inversity, or diffusion centrality) is qualitatively consistent regardless of which network is used. As an alternative to testing the consistency of a network metric, methodological work by \cite{koster2018family} proposes a generalized linear mixed model for distinguishing household and individual effects in social networks. Employing this model as an evaluation of this criteria to statistically evaluate the extent to which individual intra-household relationships are shared between household members may indicate when aggregating connections within households is a representative assumption. In the case of inconsistent metrics, we recommend analysis of the individual network. 

\subsection{Examples} \label{subsec:rec_ex}
In this section we consider three examples for how to assess each of the entitativity criteria in practice. To do so, we utilize the decision tree in \Cref{fig:dec_chart} to determine which network to study in the contexts of the experiments considered in \cite{banerjee2013}, \cite{alexander2022algorithms}, and \cite{airoldi2024}. We provide relevant context for each experiment in \Cref{apx:emp_rev} and evaluate each criteria in detail in \Cref{fig:ex_decision}. Notably, the networks we recommend studying differ from the networks that were actually studied in \cite{banerjee2013} and \cite{airoldi2024}. However, in the context of \cite{alexander2022algorithms}, our recommendation is aligned with the network studied in the original research, where the individual female network was studied.

For the setting and intervention context of \citet{banerjee2013}, in which an unweighted household network was studied in the original work, we instead recommend studying the individual network with weighted edges and with interhousehold edges only present if those connections were specifically reported. Although we determine that household members share a \textit{common fate} in this setting (the success of the intervention is tracked at the household level), the absence of \textit{similarity} and \textit{internal diffusion} between household members ultimately determines our recommendation to study a weighted individual network. In this context, it is not necessarily true that household members are similar in their financial decision-making roles \citep[e.g.,]{holvoet2005impact}, and thus we recommend weighting edges between individuals who do hold that primary household responsibility more heavily than between individuals who do not make financial decisions. As such, because not all household members necessarily share the same incentives or interest in this topic, we suggest that there is not necessarily internal diffusion in this context, and thus we recommend only connect individuals in the same household if those edges were individually reported.

In the context of \citet{airoldi2024}, where the intervention and outcome tracked is the spread of information about maternal health, we recommend studying a household network with edges weighted according to the gender of relationships between individuals in each household. This recommendation is in contrast to the unweighted household network studied in the original work. In this context, where outcomes are measured at the household level, we assess that household members share a \textit{common fate}. Furthermore, because the intervention occurs at the household level with health information presented to every household member, it is reasonable to assume that there is \textit{internal diffusion} and thus that households should be connected if any individual in one is connected to any individual in another \citep[this is how households are connected in ]{airoldi2024}. However, household members are not \textit{similar} in this context: there are likely gendered differences in how individuals discuss and share gendered health information \citep[see, e.g.,][for a discussion of how gender plays a role in family-planning decisions, discussions, and outcomes in Honduras]{berti2015adequacy}. As such, we recommend weighting the household network according to how many female-to-female connections there are between households. Under this recommendation, relationships between households where only men have relationships with other men will have a lower weight than relationships between households where women are connected to other women. This recommendation accounts for the relative differences in how gendered health information like maternal health is shared.

In each of the above experimental settings, a key goal of the work is to understand the spread of information over a network. As such, and as we saw above in \Cref{subsubsec:emp_choice}, the choice of \textit{how} to represent the network is necessary to properly understand the relevance of the conclusions which are drawn from network analysis. Our recommendations provide a principled approach for making this choice. Although out of scope for the present work, reconsidering the analysis of the above experiments with the networks we recommend, and exploring how the resulting conclusions differ is subject for interesting and important future work.

%%%%%%%%%%%%%%%%%%%%%%%%%%%%%%%%%%%%%%%%
%          LANDSCAPE TABLE             %
%%%%%%%%%%%%%%%%%%%%%%%%%%%%%%%%%%%%%%%%
\newgeometry{margin=2cm} 
\begin{landscape}
\begin{figure}
    \centering
    \includegraphics[width = \linewidth]
    {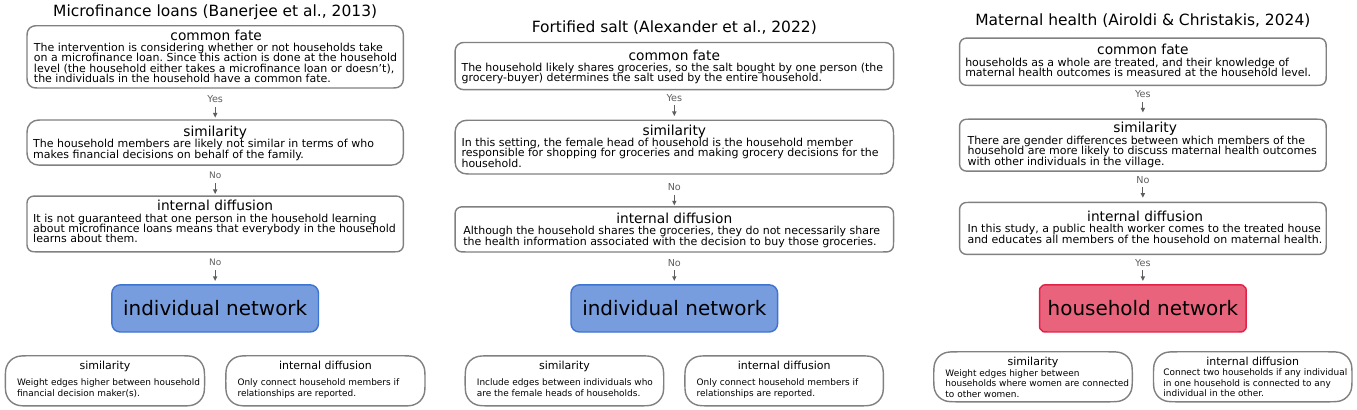}
    \caption{In this figure we show three examples \citep{banerjee2013, alexander2022algorithms, airoldi2024} of how to use the decision chart given by \Cref{fig:dec_chart} to determine what level of  network to use in a given context. For each example, we recommend a different type of network, with varying edge weights and intrahousehold connections. Our recommendations differ from the network analyzed in the original studies for both \citet{banerjee2013} and \citet{airoldi2024}.}
    \label{fig:ex_decision}
\end{figure}
\end{landscape}
\restoregeometry
%%%%%%%%%%%%%%%%%%%%%%%%%%%%%%%%%%%%%%%%
%      End LANDSCAPE TABLE             %
%%%%%%%%%%%%%%%%%%%%%%%%%%%%%%%%%%%%%%%%

\section{Conclusions} \label{sec:conclusion}
In this work we stress that there is a distinction\textemdash and a subsequent choice\textemdash between analyzing social networks at the individual and household level. In addition to emphasizing and formalizing this choice, we motivate their difference in the context of related work ranging from anthropology to political science to sociology. We explore the consequences of the choice between each network with range of empirical examples wherein analysis (local network metrics, information diffusion, and node centrality) of the network substantively change depending on which network the analysis is conducted on. 

We specifically draw upon the literature of \textit{household strategies} \citep[e.g.,][]{werner1998, wallace2002household, niehof2011conceptualizing}, which considers the role of household networks and critiques of studying household exchange relationships from an ethnographic framework. Notably, we rely on this literature for drawing our attention to the gendered labor that goes into creating and maintaining household relationships, as well as the fact that the concept and interpretation of a \textit{household} varies greatly across different cultures and economies. We employ the observations from this literature to consider that gender, power, and hierarchy should be evaluated in the context of an experiment to evaluate how edges in the resulting network should be weighted. To explore how gender impacts network structure, we construct male and female household networks from the networks in \cite{banerjee2013}, only connecting two households if a male in one is connected to a male in the other, or a female in one to a female in the other, respectively. We find that the centralities of each household differs across these gendered networks, highlighting one possible consequence of misspecifying the relevance of gender between households.

We provide a systematic recommendation for assessing which network should be analyzed in a given experimental and cultural setting grounded in the theoretical work of \cite{campbell1958common} and his criteria for \textit{entitativity}. We provide a dependent hierarchy of these entitativity criteria in the network setting by defining a decision tree based on their evaluation with respect to an experimental setting. We provide three examples for how to use this decision tree in the context of the work of \cite{banerjee2013}, \cite{alexander2022algorithms}, and \cite{airoldi2024}, ultimately recommending analysis on the weighted individual, gendered individual, and gendered household networks, respectively.

We hope that this work provides a clear opportunity for future research on networks to be more contextually-rigorous. In addition to providing specific recommendations for future network research scientists to employ, we hope that by emphasizing a difference between the individual and household networks a line of interesting research questions open. 
 
First, it is interesting to explore if the recommendations we make regarding the choice of network to study would substantively change the conclusions of the original experiments. For example, the work of \cite{banerjee2019} is focused on assessing the correlation between the individuals that villages \textit{nominate} as good gossips and the households with the highest diffusion centrality. The authors conclude that peoples' accurate perceptions of highly central people can be well-modeled with their model of diffusion centrality. It is interesting future work to then consider \textit{which network} these perceptions (and the spread of gossip) are occurring on\textemdash that is, if people's gossip nominees are even \textit{more} correlated with the individuals with highest diffusion centrality as measured on a more particular network, e.g., an individual network with edges weighted according to gender. 

In our discussion of the entitativity criterion of \textit{proximity} in \Cref{subsec:entitativity_rec}, we mentioned the additional difficulty in assessing this criterion in the case of longitudinal network studies or in cultures where the household unit is more ambiguously defined. Rigorous recommendations for how to define the boundaries of a household when its members do not remain consistent over time has been meaningfully addressed in, e.g., \cite{steele2019modeling}. Determining how to define a household \textit{network} in a similar setting of changing household memberships and boundaries is an avenue for practical and impactful future work. Furthermore, as mentioned earlier, we did not consider \citeauthor{campbell1958common}'s criterion of assessing a household's \textit{resistance to intrusion}. In many respects, we can consider the very act of network data collection as an act of intrusion. Scientists have long considered impacts of \textit{the observer effect}: how the collection of data from a system impacts the system itself. With respect to social data, anthropologists center the consideration of how the observer effect can either invalidate, bias, or strengthen ethnographic research findings \citep{lecompte1982, monahan2010}. Considering how the intrusion of social network data collection impacts the communities being researched\textemdash and the subsequent data\textemdash and incorporating households' reactions to that intrusion in evaluating entitativity, is an interesting direction for future work. 

We hope that by detailing the choices, consequences, and recommendations for studying individual and household networks, this work serves to encourage a rigorous chain of questions about the data decisions that we make when we aim to use networks as tools for exploring experimental questions about communities and the social ties between the people within them.

\section*{Acknowledgments and Disclosure of Funding} 
IA acknowledges support from the Omidyar Postdoctoral Fellowship, NSF GRFP, and the Knight-Hennessy Scholars Fellowship. PSC acknowledges support from the NSF under award \#DMS-2407058. JU acknowledges partial support from the ARO MURI award \#W911NF-20-1-0252 and the NSF under award \#2143176. We would like to thank Eleanor Power, Justin Weltz, and Kerice Doten-Snitker for helpful conversations and feedback.

\clearpage
\bibliography{mybib}
\bibliographystyle{plainnat}

\appendix
\section{Node contraction on an ~\ER~random graph}\label{apx:ER}
In \Cref{subsub:ER}, we show how when contracting random node sets of unequal size in a network generated from an \ER model, the resulting network is no longer \ER. This is not the case when contracting random node sets of equal size. 

\begin{proposition}
    Consider a random graph $G(V,E)$ generated by an \ER~model with $n$ nodes where the probability of an edge between two nodes existing occurs independently with probability $p$. Consider a corresponding unweighted household network $G'(V', E')$ constructed by randomly contracting $R$ disjoint node sets $H_1, H_2, \ldots, H_R$ of size $\ell$ where $n = R\ell$, according to the basic household contraction \Cref{eq:contraction}. Then $G'$ is also described by an \ER~random graph, $G'(R, 1 - (1-p)^{\ell^2})$.
\end{proposition}

\begin{proof}
    In the graph $G$, for any pair of nodes $v_i$ and $v_j$, edge $(v_i, v_j)$ does not exist with probability $(1-p)$. Now consider node sets $H_q$ and $H_r$. Then for node $v_i \in H_q$, there is probability $(1-p)^\ell$ that there are zero edges between node $v_i$ and any node in set $H_r$. Since there are $\ell$ nodes in set $H_q$, the total probability of no edge existing between any node in $H_q$ and any node in $H_r$ is $(1-p)^{\ell^2}$. 
    Because all edges in $G$ exist with independent probability, the probability of edges between any two sets $H_q$ and $H_r$ are also independent. For the graph $G'$ where any node $w_k$ is constructed by contracting nodes in set $H_k$, then for any pair of nodes $w_q$ and $w_r$ in $G'$, edge $(w_q, w_r)$ exists with probability $1 - (1-p)^{\ell^2}.$ Thus we can describe $G'$ as an \ER~random graph $G'(R, 1 - (1-p)^{\ell^2}).$ 
\end{proof}

\section{Detailed review of referenced network studies}\label{apx:emp_rev}

\paragraph{Diffusion of microfinance.} \cite{banerjee2013} is an exploration of how information about a microfinance loan (and subsequent uptake of a loan product) travels through villages in Karnataka, India. The study is interested in developing a model for diffusion and using this model to measure \textit{diffusion centrality} of individuals in the village. They compare the centrality of the village leaders who were first targeted with information about the microfinance loans with the eventual adoption of the loans in households across the village. Considering this study in the context of the characterizations we set above, we see that in this work: (i) the intervention (leaders being targeted) occurs at the individual level, the outcome (loan uptake) occurs at the household level, and the research question occurs at both the individual and household level; (ii) the network data is collected by taking the union of 12 separate individual relationship name generators and household networks are created by aggregating nodes and their individual edges; (iii) the authors study the household networks. In \Cref{sec:consequences}, we explore the individual and household level data that was collected across 72 villages. It is important to note that in this work, individuals in the same household are completely connected with one another\textemdash that is, all individuals living in the same house have an edge between one another. 

\paragraph{Iron-fortified salt.} \cite{alexander2022algorithms} considers the network of women living in shared apartment buildings in Mumbai, India and how information travels between them. For each of 50 buildings, women are asked to report other female heads of households in the building who they are connected to, across five different name generators. The researchers tell a subset of the women about a particular type of fortified salt available alongside its health benefits and track how information about the salt travels throughout the network and which non-targeted women end up buying the product and understanding its health benefits. Here, (i) the intervention and outcome happen at the individual level (female head of household), (ii) the network data is collected by taking the union of 5 separate name generators where only female heads of households were considered, (iii) the authors study the individual network of female heads of households. 

\paragraph{Maternal and child health.} In a large-scale experiment, \cite{airoldi2024} measure if there is more information spillover on social networks when using one-hop seeding\footnote{For the one-hop seeding, a household is selected at random, then an individual in the household is selected, one of their individual neighbours is selected at random, and then that person's household is targeted.}\citep{cohen2003efficient} as compared to other seeding strategies. Although we do not use any data from this work, we do use it as an example in \Cref{sec:rec} and \Cref{fig:ex_decision}, and thus introduce the experiment here. For this work, network data from 176 villages in Honduras was collected, and, under different targeting techniques, households were picked to receive 22 months of counseling about maternal and child health. Each month, a trained worker would visit the targeted household and provide 1-2 hours of education about maternal health. At both the beginning and end of the study, each household in the village took a comprehensive test meant to capture how their knowledge of maternal and child health outcomes changed throughout the study. Here, we see that (i) the intervention and outcome both happen at the household level, with the information being delivered to the entire household and the entire household's knowledge being tested;(ii) the network data is collected at the individual level by taking the union of three separate name generators; (iii) the network studied is, interestingly, an overlap of the individual and household level, where network metrics were reported on the household level but one-hop targeting was done on a mixture of the individual and household network, similar to the two-stage randomization discussed in \cite{basse2018analyzing}. Although the intervention took place at the household level, at a practical level there was no enforcement of every household member being present for the educational sessions, and the sessions were delivered to whichever household member was home at the time. 

\section{Further explorations into individual and household inversity} \label{apx:inversity}

In this section we expand on possible explanations for the observation noted in \cite{kumar2024friendship}: that the \textit{inversity} when computed for the individual and household networks from \cite{banerjee2013} substantively differs. As mentioned in \Cref{subsub:local}, we find that the difference in the inversity between the household and individual networks can be largely attributed to the fact that in the individual networks in the \cite{banerjee2013} dataset, individuals living in the same household are completely connected to one another. 

In fact, in \cite{kumar2024friendship}, they note that ``We can expect hub-based (or star) networks to have high inversity \dots In contrast, we expect that networks of clusters of various sizes will have low inversity.'' With this intuition in hand, and knowing that the individual networks are quite literally constructed with clusters of completely connected nodes---which get contracted when forming the household networks---the change in inversity which was surprising becomes highly expected. 

To isolate the impact of intrahousehold edges on inversity, we decompose the adjacency matrix so that we may vary the weight we give to intrahousehold connections in different contexts. This is the same decomposition we discuss in more detail in \Cref{subsec:inf_dif}.
\begin{align}
    \A^*_p = \A_{extra} + (1- p) \A_{intra},
\end{align}

\begin{figure}
    \centering
    \includegraphics[width=0.7\linewidth]{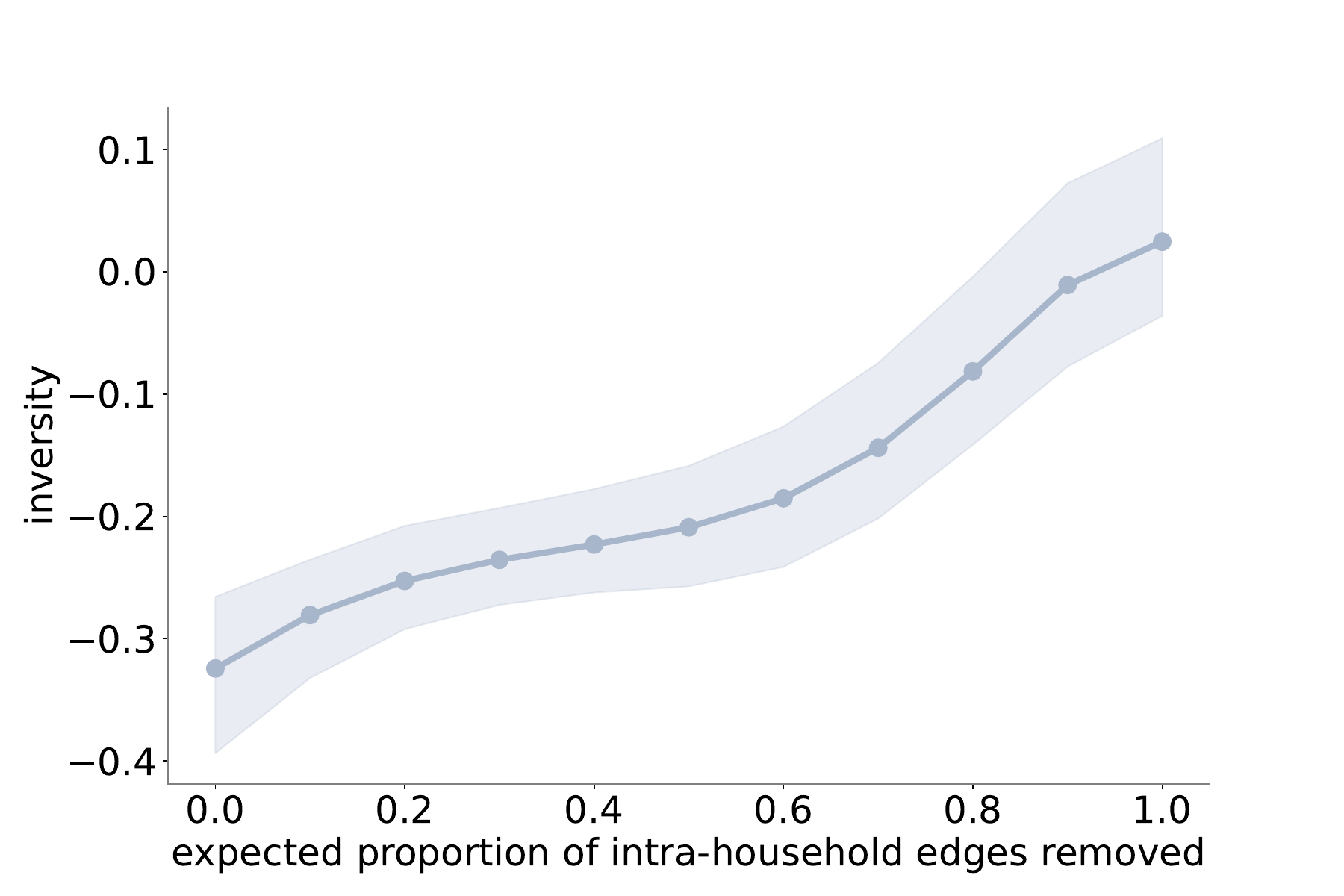}
    \caption{Expected inversity of the individual village networks with intrahousehold edges removed with probability $p$, with $p$ ranging from $0.1$ to $1$, corresponding to adjacency matrix $\A^*_p$. For each village we compute the mean inversity over 100 instances of intrahousehold edge removal, and plot the mean inversity over all villages with the shaded region capturing the 5\% and 95\% quantiles across the mean inversity for all villages.}
    \label{fig:inv_p_dots}
\end{figure}

In \Cref{fig:inv_p_dots} we plot the inversity of the individual village networks with intrahousehold edges removed with probability $p$, with $p$ ranging from 0.1 to 1. We see that inversity monotonically increases with more intrahousehold edges removed. We conclude the exploration into the impact of intrahousehold edges on inversity by considering that in \cite{kumar2024friendship}, the authors give the intuition that inversity is at its most negative when edges connect nodes of the same degree. This intuition gives more context to the impact that the complete subgraphs in the individual network have on the observed inversity\textemdash when individuals are completely connected with their household members, there is less degree variation across edges. We explore this impact further in \Cref{fig:prop_edges}, where we plot the inversity of the individual network against the proportion of the total edges in the network that exist between individuals in the same household. Indeed, we see that networks with a higher proportion of intrahousehold edges have a more negative inversity.

\begin{figure}
    \centering
    \includegraphics[width=0.7\linewidth]{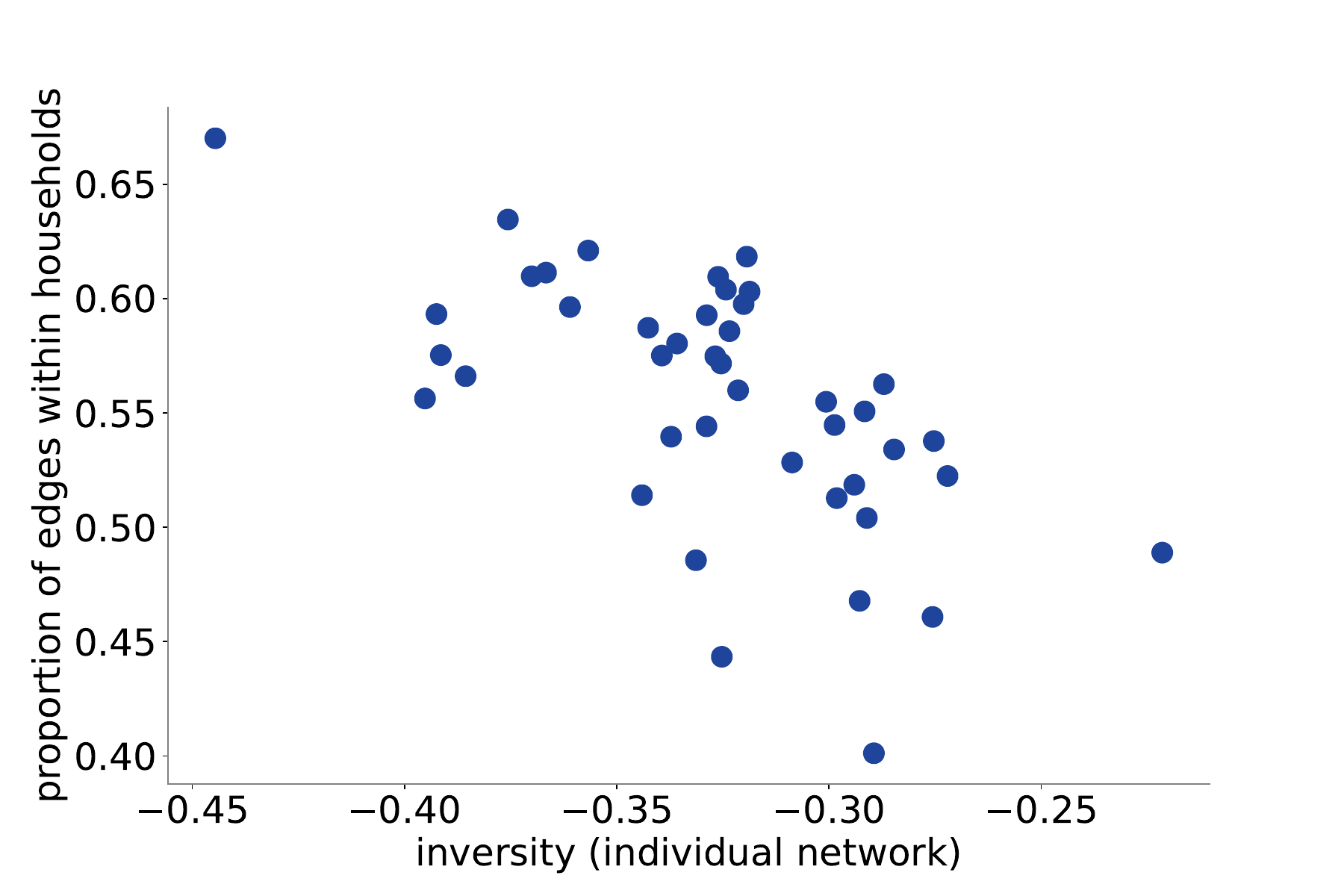}
    \caption{Individual network inversity (x-axis) against the proportion of the total edges in the network that can be attributed to edges within households (y-axis). Networks with a higher proportion of edges due to intrahousehold connections have a more negative inversity, due to less degree variation across edges.}
    \label{fig:prop_edges}
\end{figure} 

\end{document}